\documentclass[a4paper]{article}
\usepackage{geometry}
\usepackage[utf8]{inputenc}
\usepackage[english]{babel}
\usepackage{amssymb,amsmath,amsthm}
\usepackage{xcolor}
\usepackage{resizegather}
\usepackage{bbding}
\usepackage{hyperref}
\usepackage{tikz}
\usepackage{caption}

\definecolor{mygray}{rgb}{0.2, 0.2, 0.2}

\newtheorem{theorem}{Theorem}[section]

\newtheorem{lemma}[theorem]{Lemma}
\newtheorem{proposition}[theorem]{Proposition}
\theoremstyle{definition}
\newtheorem{definition}{Definition}[section]
\theoremstyle{remark}
\newtheorem{remark}{Remark}[section]
\newtheorem{example}{Example}[section]
\providecommand{\keywords}[1]{\textbf{\textbf{Keywords:}} #1}
\DeclareMathOperator{\Ric}{Ric}

\title{Notes on a conformal characterization of $2$-dimensional Lorentzian manifolds with constant Ricci scalar curvature}
\author{Nicol\`o Cangiotti{$^{1}$}, Mattia Sensi{$^{2}$}\\[1em]
\small $^1$University of Pavia, Department of Mathematics, via Ferrata 5, \\ \small 27100 Pavia (PV), Italy. Email: \texttt{nicolo.cangiotti@unipv.it}\\
 \small $^2$University of Trento, Department of Mathematics, via Sommarive 14,\\ \small 38123 Trento (TN), Italy. Email: \texttt{mattia.sensi@unitn.it}}

\date{}

\usepackage{graphicx}

\begin{document}

\maketitle

\begin{abstract} 
We present a characterization of $2$-dimensional Lorentzian manifolds with constant Ricci scalar curvature. It is well known that every $2$-dimensional Lorentzian manifolds is conformally flat, so we rewrite the Ricci scalar curvature in terms of the conformal factor and we study the solutions of the corresponding differential equations. Several remarkable examples are provided.
\bigskip

\end{abstract}

\section{Introduction}
Pseudo-Riemannian manifolds have historically been of interdisciplinary interest, and they have been extensively studied in many fields, from differential geometry to mathematical physics. In particular, there is a class which, for many reasons, has assumed an importance comparable to the one of Riemannian manifolds. It is indeed well known that, in the general theory of relativity, the spacetime has historically been modeled as a $4$-dimensional Lorentzian manifold. Lorentzian manifolds represent a very interesting case study, which has been thoroughly investigated in mathematics \cite{Beem2017,Besse2007,Chen2011} as much as in physics \cite{HawkingEllis, Pfaffle2009, Wald}. And, as it often happens in geometry, there are many properties which depend completely on the dimension of the manifold.

This paper fits into the body of works that propose an analysis of the $2$-dimensional Lorentzian geometry (see, for instance, \cite{Burgos2013,Collas77,Grumiller2002,Papa19,Terek2016,Wald91}). Our idea is the following: starting from the fact that every $2$-dimensional Lorentzian manifold is conformally flat, we show a very simple way to construct Lorentzian manifolds with constant Ricci scalar curvature by deriving different expressions of the conformal factor. In particular, we provide explicit forms for the conformal factor, which is a solution of a classical partial differential equations. In the same way, we propose also several examples, some that recover very familiar Lorentzian manifolds, whereas other are provided to illustrate the potential of our method.

The paper is organized as follows. In Section 2, we recall some definitions and a classical result regarding $2$-dimensional Lorentzian manifolds. Section 3 is devoted to illustrate a way to construct Lorentzian manifolds starting from the conformal factor. Finally, we conclude with Section 4 by summarizing the ideas exhibited and by presenting possible further developments. 

\section{Definitions and basic concepts}
In this section, we recall some basic definitions and concepts regarding Lorentzian manifolds; we remark similarities with the Riemannian case. Moreover, we state a fundamental theorem on the conformal structure of the $2$-dimensional Loretzian manifolds.
\begin{definition}
A \emph{pseudo-Riemannian manifold} $\mathcal{M}$ is a differentiable manifold, equipped with an everywhere non-degenerate metric $g$. We denote a manifold provided with a specific metric with the couple $(\mathcal{M},g)$.
\end{definition}
It is well known that, by Sylvester's law of inertia, one can identify the metric with its signature \cite{Norman1986}. We restrict our study to a particular signature, which, as mentioned above, is very relevant in many physical context, especially in the field of the general theory of relativity. 
\begin{definition} A \emph{Lorentzian manifold} is a pseudo-Riemannian manifold $\mathcal{M}$, equipped with a metric $g$ with signature $(1,n-1)$, where $n$ is the dimension of $\mathcal{M}$.
\end{definition}
For our purpose, we shall consider the $2$-dimensional Lorentzian manifolds (it is appropriate to underline that the definitions below can easily be extended to the manifolds higher dimension). As does the Euclidean space in the case of Riemannian manifolds, a key role is played by the \emph{Minkowski} space.
\begin{definition}
\label{Minkowski}
Let us consider the following metric tensor $\eta$:
\[
\eta=\begin{pmatrix}
-1 & 0\\
0 & 1
\end{pmatrix}.
\]
A Lorentzian manifold $\mathcal{M}$, equipped with the metric $\eta$, is the \emph{Minkowski space}. The metric $\eta$ is called \emph{Minkowski metric}.
\end{definition}
\begin{remark}
Historically, the definition of Minkowski space in the literature refers to the $4$-dimensional space \cite{HawkingEllis}, but this little abuse of notation should not cause confusion, as we will work in the $2$-dimensional case, unless explicitly stated otherwise. 
\end{remark}
The next definition introduce the notion of \emph{flat manifold}, by using the very usual concept of \emph{Ricci scalar curvature}.
\begin{definition}
A manifold $(\mathcal{M},g)$ with a null Ricci scalar curvature is called \emph{flat manifold}. In this case, we shall call its metric \emph{flat metric}. A manifold $(\mathcal{M},\tilde{g})$ is \emph{conformally flat} if its metric can be expressed, via the composition with a smooth function $\Omega$, as
\[
\tilde{g}=\Omega \cdot g,
\]
where $g$ is a flat metric. 
\end{definition}
\begin{remark}
It is clear that the Minkowski space of Def. \ref{Minkowski} is a flat manifold.
\end{remark}
In the previous definitions, we implicitly express the metric in terms of a tensor (actually a matrix, in this case) as in Def. \ref{Minkowski}. However, it is possible to give a different, and very useful, characterization for the metric, namely the \emph{line element}.
\begin{definition}
Let us consider a metric $g_{ij}$, where we indicate rows and columns of the matrix associated to $g$ by the indices $i$ and $j$, respectively. By considering the general curvilinear coordinates, it is possible to define the line element $ds^2$ as:
\[
ds^2=g_{ij}dq^idq^j, \quad i,j=1,2.
\]
Here we are using the Einstein summation convention.
\end{definition}
\begin{remark}
For the metric $\eta$ of Def. \ref{Minkowski} the line metric for the coordinates $(t,x)$ could be expressed by
\[
ds^2=-dt^2+dx^2
\]
\end{remark}
The fundamental result which inspired this work is given by the following theorem, which gives a very powerful strategy to handle $2$-dimensional Lorentzian manifolds.
\begin{theorem}
\label{bigtheorem}
Every $2$-dimensional Lorentzian manifold $(\mathcal{M},g)$ is conformally flat.
\end{theorem}
For the proof see, e.g., Theorem 7.2 in \cite{Schottenloher2008}. The next section is devoted to exploring the conformal factor in the case of manifolds with constant Ricci scalar curvature.

\section{Conformal characterization}
Firstly, we fix some important notation that we shall employ in this section. For the conformal factor $\Omega$, we shall make explicit the dependence on  variables, which for historically and interpretative reasons we identify with $(t,x)$.
\smallskip

Let us consider a generic $2$-dimensional Lorentzian manifold $(\mathcal{M},g)$, whose metric, thanks to Theorem \ref{bigtheorem}, can be written as
\begin{equation}
g=\Omega(t,x) \cdot \eta
\end{equation}
where $\Omega(t,x)>0$ is a smooth function and $\eta$ is the Minkowski metric defined in Def. \ref{Minkowski}. The Ricci scalar curvature $R$ can be expressed as a function of $\Omega$ and its partial derivatives by
\begin{equation}
\label{RicciOmega}
R=\frac{-\left ( \partial_t\Omega(t,x)\right )^2+\left (\partial_x\Omega(t,x) \right )^2+\Omega(t,x) \left [  \partial_t^2\Omega(t,x) - \partial_x^2\Omega(t,x) \right ]}{\left ( \Omega(t,x) \right)^3}.  
\end{equation}
If we restrict ourselves to a traditional form for $\Omega$, namely
\[
\Omega(t,x)=e^{\omega(t,x)},
\]
then Eq. \eqref{RicciOmega} becomes:
\begin{equation}
\label{RicciExp}
R= \left( {\partial_t^2}\omega(t,x) - {\partial_x^2}\omega(t,x) \right ) \mathrm{e}^{-\omega(t,x)}.
\end{equation}
Now we point out a interesting properties of the $2$-dimensional Lorentzian manifolds.
\begin{definition}
Let $(\mathcal{M},g)$ be a $2$-dmensional Lorentzian manifold. We call it a \emph{Einstein manifold} if the following condition holds:
\begin{equation}
\label{EinsCond}
\Ric=\kappa \cdot g,
\end{equation}
for some constant $\kappa \in \mathbb{R}$, where $\Ric$ denotes the Ricci curvature tensor (here we are using the notation without indices as for the metric).
\end{definition}
\begin{proposition}
A $2$-dimensional Lorentzian manifold is Einstein if and only if it has a constant Ricci scalar curvature $R$. Moreover the constant $\kappa$ is equal to $\frac{R}{2}$.
\end{proposition}
\begin{proof}
It is not difficult to compute the Ricci curvature tensor in terms of $\Omega=e^{\omega(t,x)}$. We get:
\[
\begin{pmatrix}
 \frac{1}{2} \left(\Omega ^{(0,2)}(t,x)-\Omega ^{(2,0)}(t,x)\right) & 0 \\
 0 & \frac{1}{2} \left(\Omega ^{(2,0)}(t,x)-\Omega ^{(0,2)}(t,x)\right)
\end{pmatrix}.
\]
Now we see that the condition given by Eq. \eqref{EinsCond} holds if and only if
\[
2\kappa e^{\omega(t,x)}=\left (-\partial_t^2\omega(t,x)+\partial^2_x\omega(t,x) \right ).
\]
Thanks to Eq. \eqref{RicciExp} we conclude the proof.
\end{proof}
The conformal characterization of the Ricci scalar curvature given by Eq. \eqref{RicciExp} provide a very simple way to construct Lorentzian manifolds with constant curvature. The first step is to obtain the general form of a flat manifold. It is trivial how this is strictly connected with the classical change of variables as we shall see in Example \ref{Ex1}.
\begin{lemma}
\label{Lemma1}
The conformal factor $\Omega(t,x)$ of a generic $2$-dimensional flat Lorentz manifold (i.e. with $R=0$) can be written as
\begin{equation}
\label{omegaflat}
\Omega(t,x)=\phi(x+t)\cdot\psi(x-t).
\end{equation}
\end{lemma}
\begin{proof}
Let we fix $R=0$. Thus, Eq. \eqref{RicciExp} leads to the one-dimensional wave equation. It is well known (see, e.g. case 8 of Sect. 4.1.3 in \cite{Polyanin2004}) that its solutions can be expressed by the sum of a right traveling function $\hat{\phi}$ and a left traveling function $\hat{\psi}$. This means we can write:
\begin{equation*}
    \omega(t,x)=\hat{\phi}(x+t)+\hat{\psi}(x-t),
\end{equation*}
and so
\begin{align*}
    \Omega(t,x)&=e^{\omega(t,x)}=\exp\left [ \hat{\phi}(x+t)+\hat{\psi}(x-t)\right ]= \\
    &=\exp\left [ \hat{\phi}(x+t)\right ]\cdot \exp\left[ \hat{\psi}(x-t) \right]={\phi}(x+t)\cdot{\psi}(x-t).
\end{align*}
\end{proof}
\begin{remark}
\label{Rmk31}
It is remarkable to notice that, setting $u=x+t$ and $v=x-t$, we have actually switched to the \emph{null coordinates}. It is trivial to prove that, starting with the metric expressed in null coordinates, namely $ds^2=dudv$, we get the same result of Lemma \ref{Lemma1}.
\end{remark}
\begin{example}
\label{Ex1}
Let we consider, for instance, the usual conformal transformation of the Minkowski space, which leads to the \emph{Penrose-Carter diagram} (see, for instance, Chpt. 5 in \cite{HawkingEllis}). It is easy to see, starting from the null coordinates\footnote{The trick is to set $\mu=\arctan(u)$ and $\nu=\arctan(v)$, with $-\frac{\pi}{2}\le \mu \le \frac{\pi}{2}$ and $-\frac{\pi}{2}\le \nu \le \frac{\pi}{2}$.}, that the conformal factor could be written in terms of $t$ and $x$ as
\[
\Omega(t,x)=\left( \cos(t)+\cos(x)\right)^{-2}=\frac{1}{4} \cos^{-2} \left(\frac{x+t}{2}\right) \cos^{-2} \left(\frac{x-t}{2}\right),
\]
where now the range of $t$ and $x$ is given by:
\[
-\pi < x < \pi \quad \text{and} \quad -\pi+|x| < t < \pi-|x|.
\]
Thus, we recover, as expected, the form of Eq. \eqref{omegaflat}.
\end{example}
\begin{figure}[ht!]
    \centering
    \begin{minipage}[b]{0.45\textwidth}
            \centering
    \begin{tikzpicture}
        \node at (0,0){\includegraphics[width=0.95\textwidth]{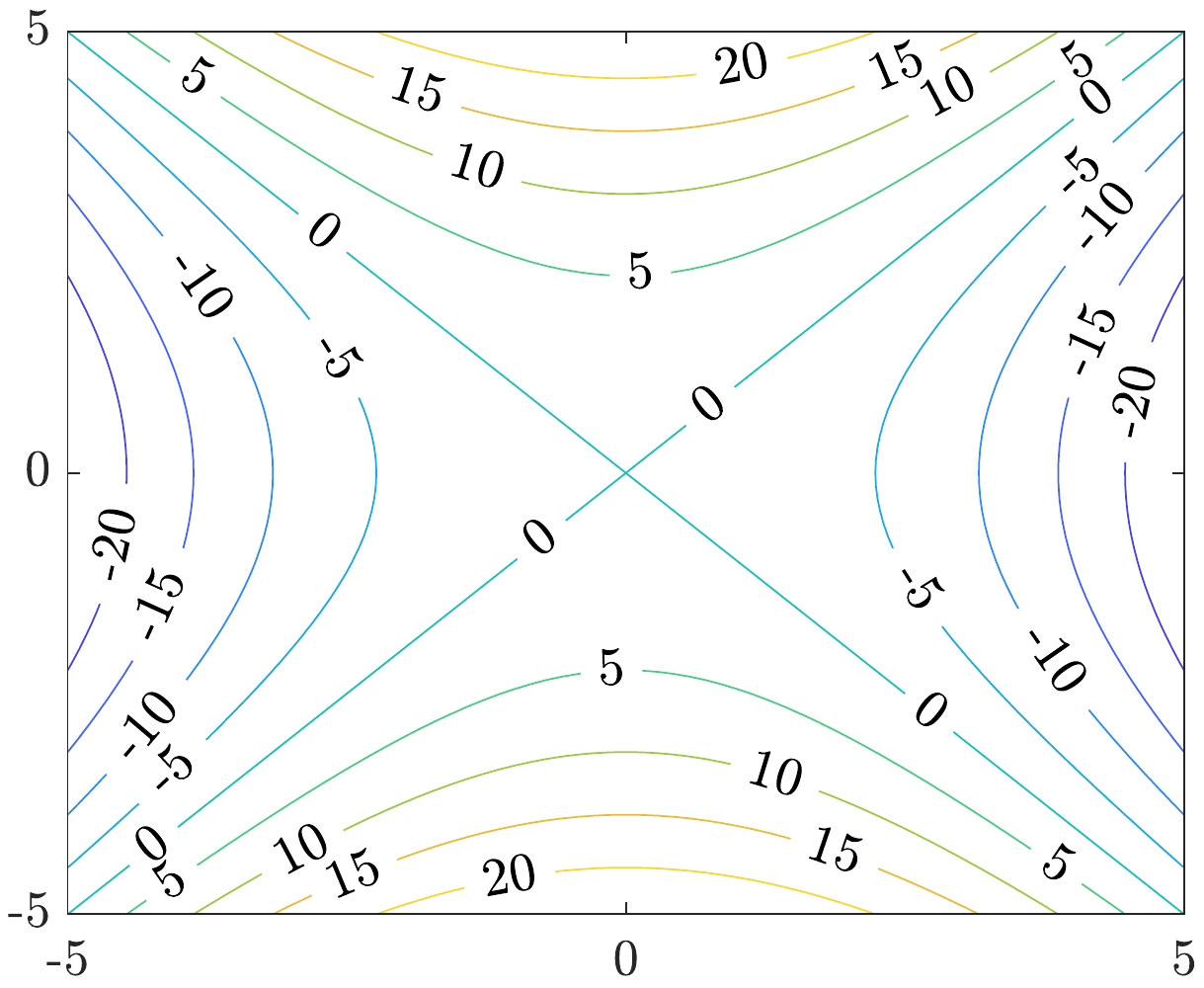}};
        \node at (0,2.8){$-t^2+x^2=s^2$};
        \node at (0,-2.8){$x$};
        \node at (-3.2,0){$t$};
     \end{tikzpicture}
  \caption*{(a) Minkowski space in classic coordinates.}
    \end{minipage}\hfill
    \begin{minipage}[b]{0.45\textwidth}
        \centering
            \begin{tikzpicture}
        \node at (0,0){\includegraphics[width=0.95\textwidth]{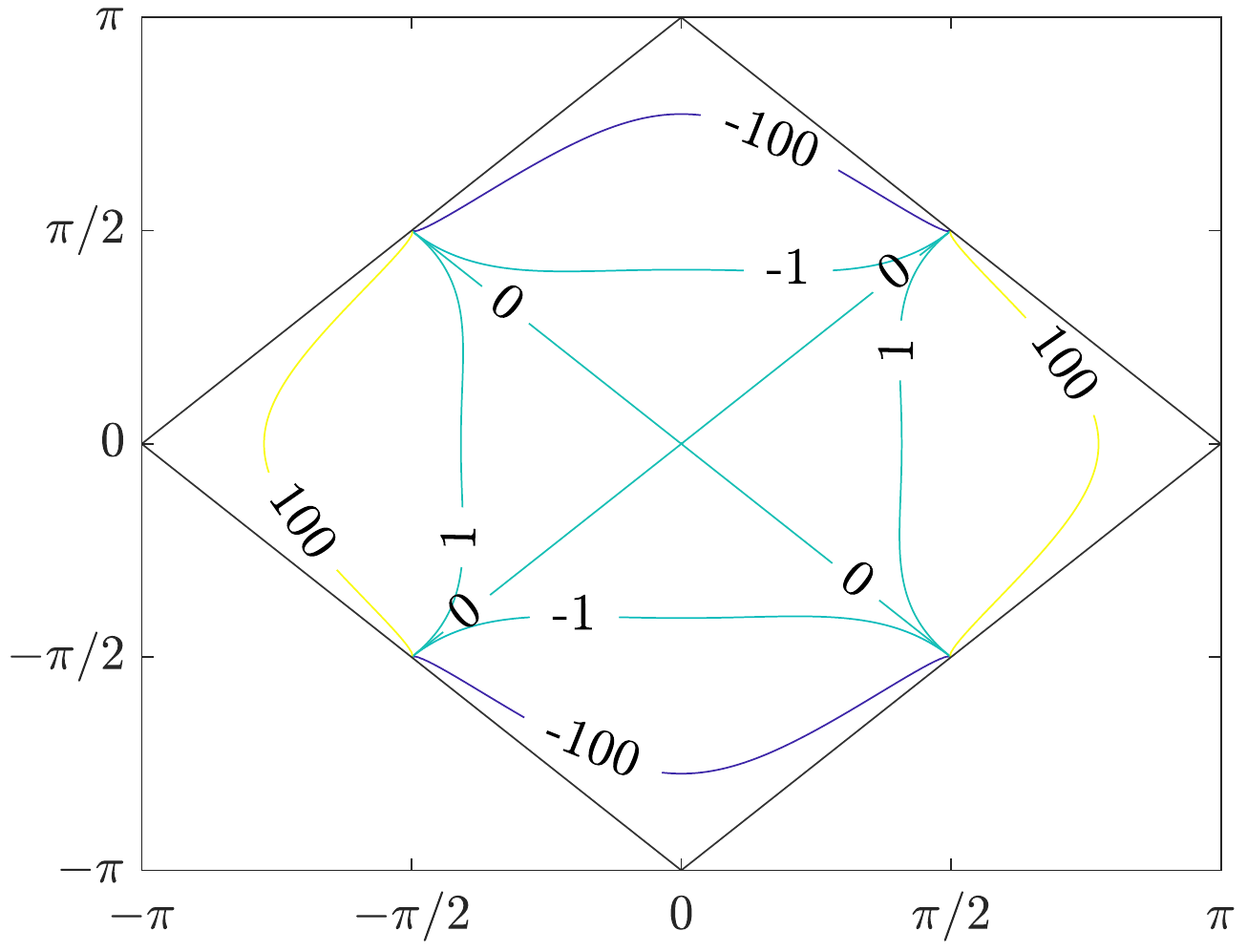}};
        \node at (0,2.7){$(\cos t+\cos x)^{-2}(-t^2+x^2)=s^2$};
        \node at (0,-2.5){$x$};
        \node at (-3,0){$t$};
        \node at (1.5,1.6){\color{mygray} $\infty$};
        \node at (-1.2,1.4){\color{mygray} $\infty$};
        \node at (1.5,-1.3){\color{mygray} $\infty$};
        \node at (-1.2,-1.1){\color{mygray} $\infty$};
     \end{tikzpicture}
        {\caption*{(b) Penrose diagram of Minkowski space.}}
    \end{minipage}
  \caption{Here we provide a comparison between the Minkowski space in normal coordinates (a) and its conformal representation as Penrose-Carter diagram (b). The figures present the locus of point with the same \emph{space interval} $s^2$ (also known in literature as \emph{spacetime interval}), for several different values. We notice that the compactification in (b) produces the classical ``diamond shape'', where the edges stand for the infinity. }
\end{figure}

Now we shift the focus on those conformal factors that depend only of one variable. We shall show (Example \ref{DeSitter}) how the general solution lead to a well known manifold: the de Sitter space.

\begin{lemma}
Let us assume $R\neq 0$. The space-independent conformal factor $\Omega(t)$ of a generic $2$-dimensional Lorentz manifold can be written as
\begin{equation}
\label{timeomega}
\Omega(t)=\frac{c_1\left (-1+\tanh^2\left ( \frac{1}{2} \sqrt{c_1(t+c_2)^2}\right ) \right)}{2R}.
\end{equation}
with $c_1$ and $c_2$ arbitrary constants.
\smallskip

Similarly, the time-independent conformal factor $\Omega(x)$ of a generic $2$-dimensional Lorentz manifold can be written as
\begin{equation}
\label{spaceomega}
\Omega(x)=\frac{d_1\left (1-\tanh^2\left ( \frac{1}{2} \sqrt{d_1(x+d_2)^2}\right ) \right)}{2R},
\end{equation}
with $d_1$ and $d_2$ arbitrary constants.
\end{lemma}
\begin{proof}
Here we just prove the space-independent case; the proof of the time-independent case is analogous. Assuming $\omega$ to be independent of $x$, Eq. \eqref{RicciExp} becomes
\begin{equation}
    \label{xindep}
\frac{\partial^2}{\partial t^2}\omega(t)=R\mathrm{e}^{\omega(t)}.
\end{equation}
After some computations, we obtain the solutions to the second order nonlinear ODE (\ref{xindep}): 
\begin{equation*}
    \omega(t)=\log\left( \frac{c_1\left (-1+\tanh^2\left ( \frac{1}{2} \sqrt{c_1(t+c_2)^2}\right ) \right)}{2R}\right),
\end{equation*}
where $c_1,c_2 \in \mathbb{R}$. This concludes the proof. 
\end{proof}
\begin{example}
\label{DeSitter}
Let us consider Eq. \eqref{timeomega}. Setting $c_1=-4$, $c_2=0$ and $R=2$, we get:
\begin{equation*}
\Omega(t)=1-\tanh^2(\sqrt{-t^2})=\sec^2(t),
\end{equation*}
which, for the range 
\[
\quad -\frac{\pi}{2}<t<\frac{\pi}{2} \quad \text{and} \quad 0\le x<2\pi,
\]
is exactly the conformal factor defining the  $2$-dimensional \emph{de Sitter space} \cite{Coxeter43} in global coordinates (i.e. $ds^2=\sec^2(t)(-dt^2+dx^2)$).
\end{example}
\begin{figure}[ht!]
    \centering
    \begin{tikzpicture}
        \node at (0,0){\includegraphics[scale=0.475]{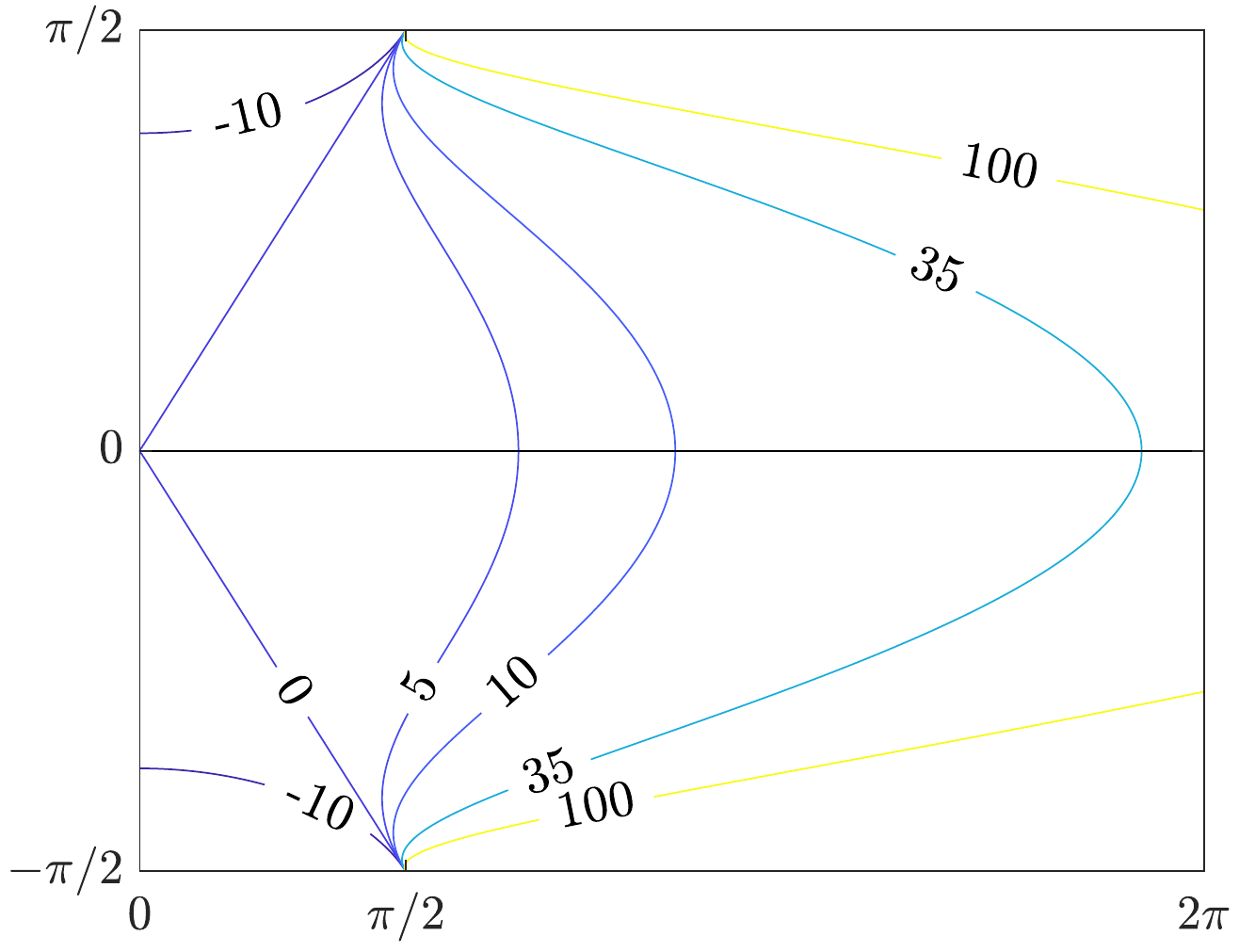}};
        \node at (0,2.6){$\sec^2 t(-t^2+x^2)=s^2$};
        \node at (0,-2.5){$x$};
        \node at (-2.8,0){$t$};
     \end{tikzpicture}
    \caption{Penrose-Carter diagram of the $2$-dimensional de Sitter space in global coordinates.}
    \label{DeSitterFig}
\end{figure}
Finally we have the general case, which allow us to construct several interesting manifolds and their Penrose- diagram as we present in Example \ref{Ex33}.
\begin{theorem}
\label{GeneralThm}
The general solution of Eq. \eqref{RicciExp} is given by
\begin{gather}
\omega(t,x) =\phi(x+t)+\psi(x-t)-2\log \left | k\int^{x+t}_{} e^{\phi(\lambda)}\textnormal{d}\lambda -\frac{R}{8k} \int^{x-t}_{} e^{\psi(\lambda)} \textnormal{d}\lambda+C \right  |,
\end{gather}
which, in the original $\Omega$ notation, becomes
\begin{gather}
\label{GeneralOmega}
\Omega(t,x) =\exp \left [\phi(x+t) \right ]\cdot\exp \left [\psi(x-t)\right ]\cdot \left ( k\int^{x+t}_{} e^{\phi(\lambda)}\textnormal{d}\lambda -\frac{R}{8k} \int^{x-t}_{} e^{\psi(\lambda)} \textnormal{d}\lambda + C\right )^{-2},
\end{gather}
where $\phi$ and $\psi$ arbitrary (smooth) functions, $k\neq 0$ and $C$ is an arbitrary real constants. With $\int ^{x}f(\lambda )d\lambda$  we denote the integration of $f(\lambda )$ with respect to $\lambda$  and then the substitution $\lambda=x$.
\end{theorem}
In literature, equation \eqref{RicciExp} is called \emph{modified} (or \emph{generalized}) \emph{Liouville equation}, and it is deeply studied in many works (for further details, see \cite{Liouville1853, Bullough1980,  Polyanin2004}).
\begin{remark}
\label{Rmk32}
As in Remark \ref{Rmk31}, we can rewrite Eq. \eqref{GeneralOmega} using null coordinates $(u,v)$ (with $ds^2=dudv$). So we get
\begin{equation}
\label{GeneralNull}
\Omega(u,v)=e^{\phi(u)}\cdot e^{\psi(v)}\left ( k\int^u e^{\phi(\lambda)} \textnormal{d}\lambda  -\frac{R}{8k}\int^v e^{\psi(\lambda)} \textnormal{d}\lambda  +C\right)^{-2}
\end{equation}
with $\phi,\psi$ and $C,k$ as in Theorem \ref{GeneralThm}.
\end{remark}
\begin{example}
\label{Ex33}
Let us consider Eq. \eqref{GeneralOmega} and let $\phi$ and $\psi$ be the identity function. We also set $k=1$, $R=2$ and $C=0$. So we obtain the following conformal factor in standard coordinates
\begin{equation*}
    \Omega_1(t,x)=e^{2x}\left (e^{x+t}-\frac{1}{4}e^{x-t} \right )^{-2},
\end{equation*}
which defines a metric of scalar curvature $R=2$; see Figure \ref{Figure_3}(a) for a visualization. By using Remark \ref{Rmk32} and Eq. \eqref{GeneralNull}, it is possible to provide the Penrose- diagram (see  Example \ref{Ex1}) as one can see in Figure \ref{Figure_3}. In this case the conformal factor is given by:
\[
\Omega_2(t,x)=\frac{\exp \left(\tan \left(\frac{t+x}{2}\right)\right) \exp \left(\tan \left(\frac{x-t}{2}\right)\right)}{\left(2 \cos \left(\frac{t+x}{2}\right) \cos \left(\frac{x-t}{2}\right)\right)^2 \left(\exp\left(\tan \left(\frac{t+x}{2}\right)\right)-\frac{1}{4} \exp\left(\tan \left(\frac{x-t}{2}\right)\right )\right)^2},
\]
where the range of $t$ and $x$ is:
\[
-\pi < x < \pi \quad \text{and} \quad -\pi+|x| < t < \pi-|x|.
\]
See Figure \ref{Figure_3}(b) for a useful visualization.
\end{example}

\begin{figure}[ht!]
    \centering
    \centering
    \begin{minipage}[b]{0.45\textwidth}
            \centering
    \begin{tikzpicture}
        \node at (0,0){\includegraphics[width=0.95\textwidth]{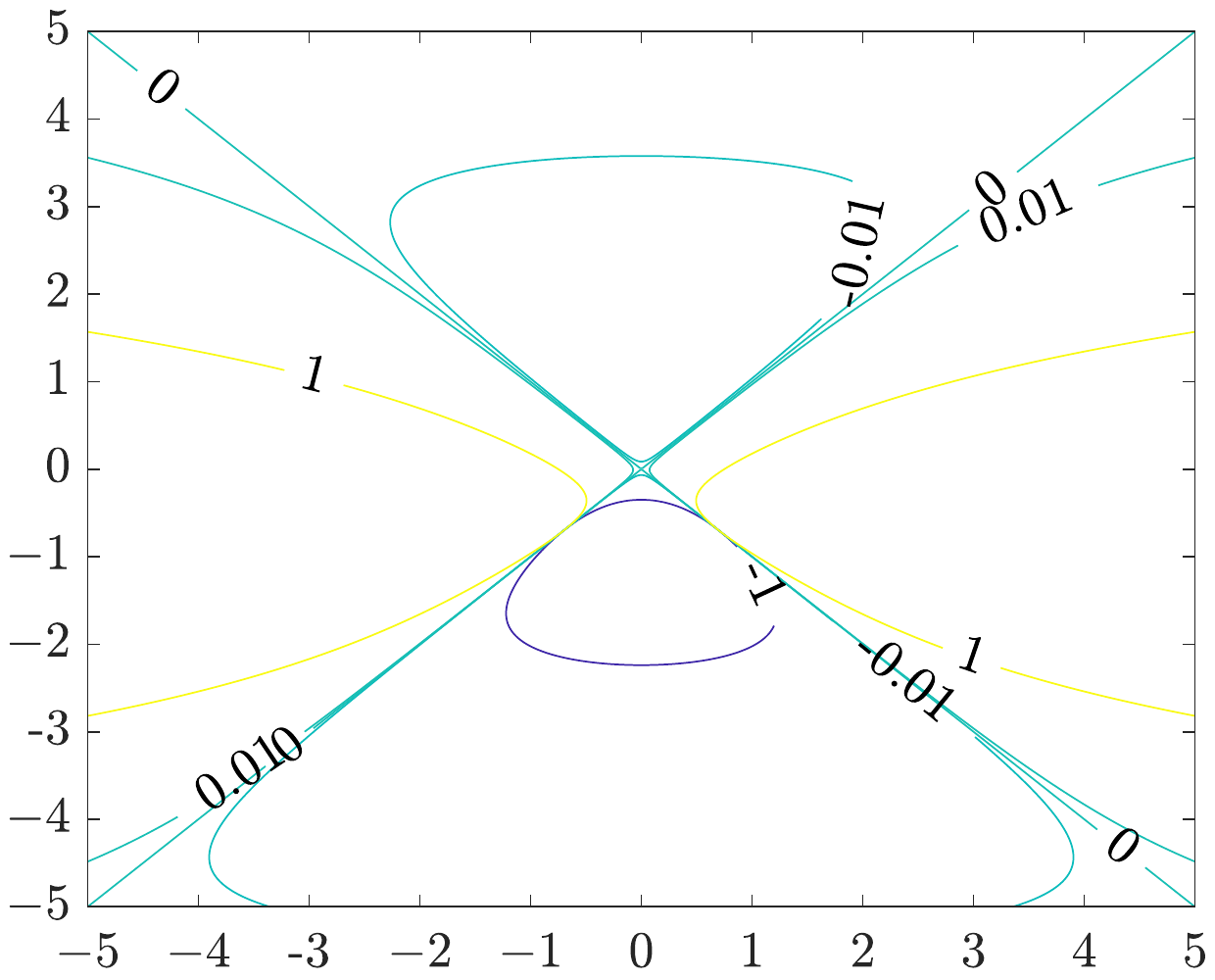}};
        \node at (0,2.8){$\Omega_1(t,x)(-t^2+x^2)=s^2$};
        \node at (0,-2.8){$x$};
        \node at (-3.2,0){$t$};
     \end{tikzpicture}
  \caption*{(a) Space of Example \ref{Ex33} in classic coordinates.}
    \end{minipage}\hfill
    \begin{minipage}[b]{0.45\textwidth}
        \centering
            \begin{tikzpicture}
        \node at (0,0){\includegraphics[width=0.95\textwidth]{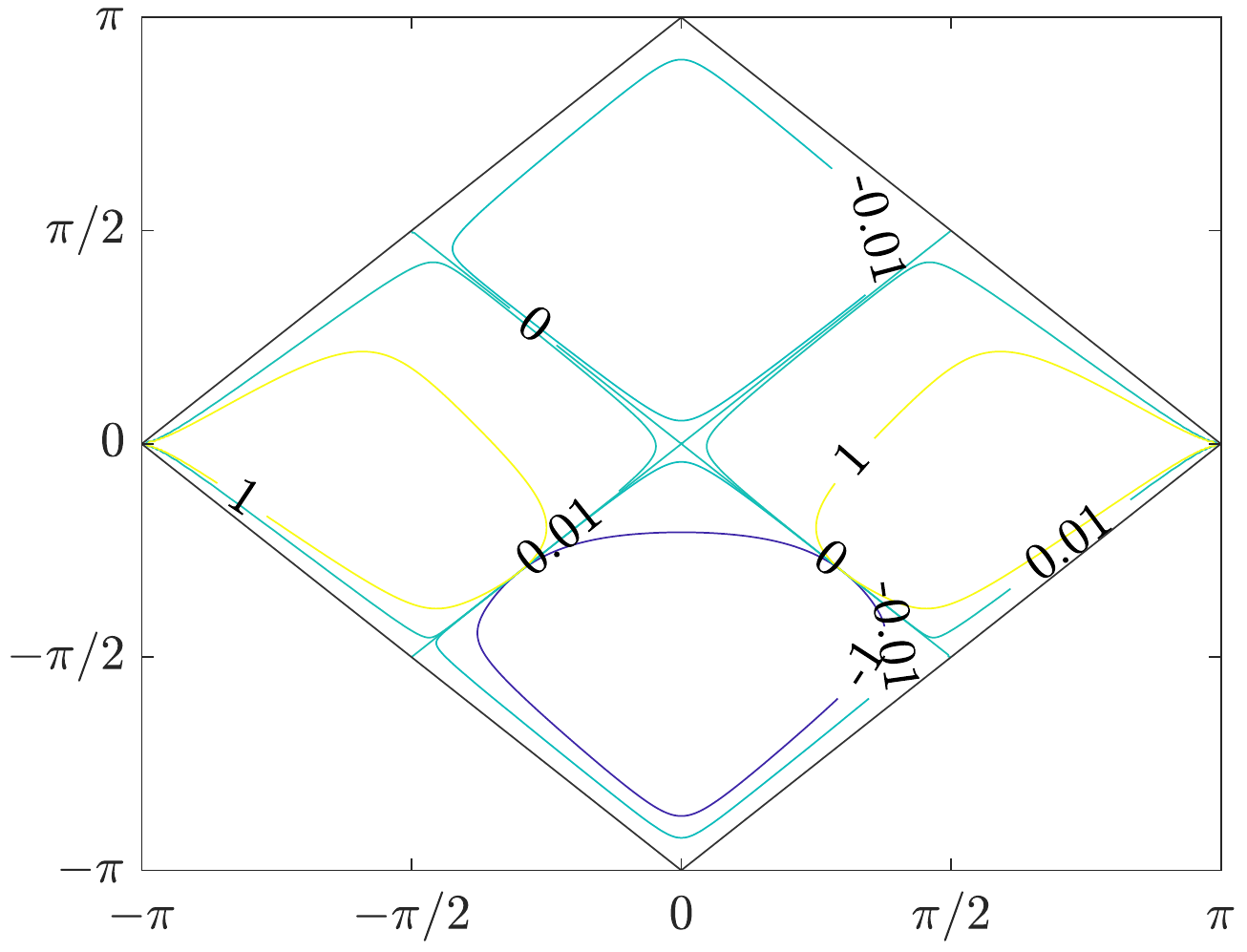}};
        \node at (0,2.7){$\Omega_2(t,x)(-t^2+x^2)=s^2$};
        \node at (0,-2.5){$x$};
        \node at (-3,0){$t$};
        \node at (1.5,1.6){\color{mygray} $\infty$};
        \node at (-1.2,1.4){\color{mygray} $\infty$};
        \node at (1.5,-1.3){\color{mygray} $\infty$};
        \node at (-1.2,-1.1){\color{mygray} $\infty$};
     \end{tikzpicture}
        {\caption*{(b) Penrose diagram of the space of Example \ref{Ex33}.}}
    \end{minipage}
    \caption{Example \ref{Ex33} lead us to the construction of the particular space (where the Ricci scalar curvature is $R=2$) shown both in standard coordinates (a) ad by the corresponding Penrose-Carter diagram (b).}
    \label{Figure_3}
\end{figure}

\begin{remark}
It is now possible, by varying parameters, to obtain all the possible Lorentzian space with constant Ricci scalar curvature and applying the same argument we used in this work, also their corresponding Penrose-Carter diagram.
\end{remark}

\section{Conclusions}
In this note, we provided a constructive way to develop a conformal characterization of $2$-dimensional Lorentzian manifolds with constant Ricci scalar curvature. The tools we used are very classical, but their application allow us to create a suggestive set of Lorentzian manifolds, starting from the conformal factor which, as we have shown, identifies them. Many useful examples have been proposed to explain the power of such construction. In the near future, we would like to investigate other similar applications using the same approach to characterize $2$-dimensional Lorentzian manifolds with a general (i. e. not constant) Ricci scalar curvature, or for the $3$-dimensional Einstein manifolds.

\par\bigskip\noindent
\textbf{Acknowledgments:} The authors would like to thank Daniele Taufer for the fruitful exchanges during the writing of this paper. NC and MS would like to thank the University of Pavia and the University of Trento, respectively,
for supporting their research.

\bibliographystyle{plain} \small
\bibliography{biblio}

\end{document}